\documentclass[runningheads]{llncs}
\usepackage[T1]{fontenc}
\usepackage{graphicx}
\usepackage{algorithm}
\usepackage[noend]{algpseudocode}
\usepackage{amsmath}
\usepackage{amssymb}
\usepackage{subcaption}
\usepackage{orcidlink}
\usepackage{xcolor}
\usepackage{array}
\newcommand{\share}[1]{[\![#1]\!]}
\DeclareMathOperator{\Gen}{Gen}
\DeclareMathOperator{\Eval}{Eval}
\DeclareMathOperator{\Decode}{Decode}

\DeclareMathOperator{\Real}{Real}
\DeclareMathOperator{\Sim}{Sim}
\DeclareMathOperator{\Leak}{Leak}
\DeclareMathOperator{\Ideal}{Ideal}
\DeclareMathOperator{\Output}{Output}

  \let\proof\relax\let\endproof\relax
\usepackage{amsthm} 

\begin{document}
\title{Eliminating Exponential Key Growth in PRG-Based Distributed Point Functions}
\titlerunning{Eliminating Exponential Key Growth in  PRG-based DPF}
%
\author{Marc Damie\inst{1,2}\orcidlink{0000-0002-9484-4460}\thanks{Corresponding author: \email{m.f.d.damie@utwente.nl}}\and
      Florian Hahn\inst{1}\orcidlink{0000-0003-4049-5354}\and
      Andreas Peter\inst{3}\orcidlink{0000-0003-2929-5001}\and
      Jan Ramon\inst{2}}

\authorrunning{M. Damie et al.}

\institute{University of Twente, The Netherlands\and
    Inria, France \and
    Carl von Ossietzky Universität Oldenburg, Germany}
\maketitle              
\begin{abstract}
Distributed Point Functions (DPFs) enable sharing secret point functions across multiple parties, supporting privacy-preserving technologies such as Private Information Retrieval, and anonymous communications.
While 2-party PRG-based schemes with logarithmic key sizes have been known for a decade, extending these solutions to multi-party settings has proven challenging.
In particular, PRG-based multi-party DPFs have historically struggled with practicality due to key sizes growing exponentially with the number of parties and the field size.

Our work addresses this efficiency bottleneck by optimizing the PRG-based multi-party DPF scheme of Boyle et al. (EUROCRYPT'15).
By leveraging the honest-majority assumption, we eliminate the exponential factor present in this scheme.
Our construction is the first PRG-based multi-party DPF scheme with practical key sizes, and provides key up to $3\times$ smaller than the best known multi-party DPF.
This work demonstrates that with careful optimization, PRG-based multi-party DPFs can achieve practical performances, and even obtain top performances.
\end{abstract}

\section{Introduction}
Function Secret Sharing \cite{boyle_function_2015} is a cryptographic primitive enabling to share secret functions.
In these protocols, a key dealer knowing a secret function $f$ distributes $p$ keys to different shareholder.
Each shareholder can use its key to obtain a share of $f(x)$, without any communication between the shareholders.

Among all function families, schemes supporting point functions (i.e., $f(x)=\beta$ if $x=\alpha$, $0$ otherwise) attracted a lot of attention thanks to their numerous applications notably in Private Information Retrieval (PIR) \cite{gilboa_distributed_2014}, in anonymous communications \cite{corrigan-gibbs_riposte_2015}, in digital currencies \cite{zyskind_high-throughput_2024}, and machine learning \cite{boyle_secure_2019}.
These schemes are called ``Distributed Point Functions'' (DPF) \cite{boyle_function_2015,gilboa_distributed_2014}.

To support these applications, there is a significant research incentive aiming to improve existing schemes, notably their key size.
DPF efficiency is commonly evaluated based on the influence of the function domain size ($N$) on the key size.
For two- and three-party DPF, schemes based on PseudoRandom Generators (PRGs) provide logarithmic key sizes \cite{boyle_function_2015,zyskind_high-throughput_2024}.
However, there is still a lot of active research to obtain similar key sizes for any arbitrary number of parties.

In multi-party DPF, three main approaches have emerged.
First, elliptic-curve-based schemes \cite{corrigan-gibbs_riposte_2015,kumar_compact_2024} offer practical $O(\sqrt{N})$ key sizes, but they require a non-linear share decoding.
This non-linearity makes them incompatible with several key applications such as PIR.
Second, Boyle et al. \cite{boyle_function_2015} presented a dishonest-majority scheme with $O(\sqrt{N})$ key size, and Bunn et al. \cite{bunn_cnf-fss_2022} an honest-majority scheme with $O(\sqrt[4]{N})$ key size.
Unfortunately, this asymptotic cost hides an exponential factors $q^p$ (for output shares in $\mathbb{F}_q$ and $p$ parties).
This factor makes these schemes impractical for any modulus $q>210$ (as detailed in Section \ref{sec:comparison}).
This problem was identified in existing works \cite{boyle_function_2022,kumar_compact_2024}, but has not been solved \emph{yet}.

Finally, Bunn et al. \cite{bunn_cnf-fss_2022} proposed a third approach based on honest-majority to build an information-theoretic (IT) scheme with $O(\sqrt{N})$ key size \emph{and no exponential factor}.
This scheme is the only multi-party scheme with practical key sizes supporting all applications.
Even though this scheme is practical, solving the efficiency issues of the other schemes could lead to even better performances.
As PRGs have lead to logarithmic key sizes in two and three-party schemes, optimizing multi-party PRG-based schemes could be promising.

\paragraph{Our Contributions}
Our paper optimizes the PRG-based DPF proposed by \cite{boyle_function_2015}.
Our optimized scheme avoids the exponential factors present in \cite{boyle_function_2015} using the honest-majority assumption; obtaining a key size of $O(\sqrt{N \cdot \binom{p}{m+1}}\log q)$ instead of $O(\sqrt{N}q^{\frac{p-1}{2}}\log q)$.
Our benchmark shows that our scheme is the first multiparty PRG-based scheme with practical key sizes.
We even provide keys up to $3\times$ smaller than the best performing DPF (i.e., the IT DPF by \cite{bunn_cnf-fss_2022}).

\paragraph{Notations}
Let $p$ be the number of parties/shareholders, $m$ be the number of corrupted parties.
Like most FSS works \cite{boyle_function_2015,boyle_function_2016,bunn_cnf-fss_2022,gilboa_distributed_2014,kumar_compact_2024}, we focus on semi-honest adversaries: follow the protocol and infer \emph{passively} secret information.

Let $\mathbb{F}_q$ be a prime field, $N$ be the function domain size, $1^\lambda$ a security parameter, and $\share{x}_i$ be the $i$-th share of the secret $x$.
Let $\nu=\lceil\sqrt{N}\rceil$ and $C=\binom{p}{m+1}$.
Let $G: \{0,1\}^\lambda \rightarrow \mathbb{G}^\nu$ be a PRG, and $\mathbb{G}$ an Abelian group.

\section{Background}
\label{subsec:def-fss}
Function secret sharing (FSS) \cite{boyle_function_2015} enables sharing secret functions between $p$ parties.
Each FSS scheme can share function from a specific function family.

A function family $\mathcal{F}$ \cite{boyle_function_2022} is a pair $(P_\mathcal{F}, E_\mathcal{F})$ where $P_\mathcal{F} \subseteq \{0,1\}^*$ is a collection of function descriptions $\hat{f}$, and $E_\mathcal{F}$: $P_\mathcal{F} \times \{0, 1\}^* \rightarrow \{0, 1\}^*$ is a polynomial-time algorithm defining the function described by $\hat{f}$; i.e., $f(x) = E_\mathcal{F}(\hat{f}, x)$.
All functions within a family share the same domain $\mathcal{X}$ and output space $\mathcal{Y}$.

Due to their applications notably in PIR \cite{gilboa_distributed_2014} and anonymous communications \cite{corrigan-gibbs_riposte_2015}, the most studied function family has been point functions \cite{boyle_function_2015,boyle_function_2016,bunn_cnf-fss_2022,corrigan-gibbs_riposte_2015,gilboa_distributed_2014}; functions $f$ such that $f(x) = \beta$ if $x=\alpha$, and $f(x)=0$ otherwise.
For point functions, the function description is the tuple $(\alpha, \beta)\in P_\mathcal{F}$.
Schemes supporting point functions are called ``Distributed Point Functions'' (DPF).

For a function family $\mathcal{F}$, we define a $p$-party FSS scheme using 3 algorithms:
\begin{itemize}
  \item $\Gen: \mathbb{N} \times P_\mathcal{F} \rightarrow \mathcal{K}^p$ takes as input a security parameter $1^\lambda \in\mathbb{N}$ and a function description $\hat{f}\in P_\mathcal{F}$, and outputs $p$ keys $k_1, \dots, k_p$.
  \item $\Eval: \mathcal{K} \times \mathcal{X} \rightarrow \mathbb{G}$ takes as input $k_i$ and a point $x\in\mathcal{X}$, outputs a share of $f(x)$ that we denote as $\share{f(x)}_i$.
  \item $\Decode: \mathbb{G}^p \rightarrow  \mathcal{Y}$ takes as input the shares $\{\share{f(x)}_1, \dots,\share{f(x)}_p \}$ and outputs the secret $f(x)$.
\end{itemize} 

\begin{definition}[Correctness \cite{boyle_function_2015}] A scheme $(\Gen, \Eval, \Decode)$ is correct if, for any function $f\in\mathcal{F}$ and point $x\in\mathcal{X}$, we have:
  \begin{align*}
    \Pr\left[\Decode(\Eval(k_1, x), \dots, \Eval(k_p, x)) = f(x)\right] = 1\\
    \text{with } k_1, \dots, k_p \leftarrow \Gen(1^\lambda, \hat{f})
  \end{align*}

\end{definition}

\begin{definition}[Privacy \cite{boyle_function_2022}]
  \label{def:privacy}
  Let $\Leak : \{0, 1\}^* \rightarrow \{0, 1\}^*$ be a function specifying the allowable leakage.
  A scheme $(\Gen, \Eval, \Decode)$ is private if, for every set of $m$ corrupted parties $S \subseteq \{1\dots p\}$, there exists a PPT algorithm $\Sim$ (simulator), s.t. for any sequence of function descriptions ($\hat{f}_1, \hat{f}_2, \dots$) of size polynomial in $\lambda$, the outputs of $\Real$ and $\Ideal$ are computationally indistinguishable:

  \begin{itemize}
    \item $\Real(1^\lambda): (k_1, \dots, k_p) \leftarrow \Gen(1^\lambda, \hat{f}_\lambda); \Output~ (k_i)_{i\in S}$
    \item $\Ideal(1^\lambda): \Output~ \Sim(1^\lambda, \Leak(\hat{f}_\lambda))$
  \end{itemize}
\end{definition}

\paragraph{DPF by \cite{boyle_function_2015}}
To build their multi-party DPF scheme under dishonest majority ($m < p$), Boyle et al. \cite{boyle_function_2015} represented the domain $\{1, \dots, N\}$ as a $\nu \times \nu$ grid (with $\nu = \lceil\sqrt{N}\rceil$).
This grid is full of zeros except on the cell $(\gamma_*, \delta_*)$, with $\alpha = \gamma_* \nu + \delta_*$.

For each row $\gamma\in\{1, \dots, \nu\}$, the $\Gen$ algorithm samples $q^{p-1}$ random $\lambda$-bit random seeds $s_{\gamma,1} \dots s_{\gamma, q^{p-1}}$.
For each seed $s_{\gamma,j}$, the algorithm generates additive shares of a coefficient $a_{\gamma,j}$: $\share{a_{\gamma,j}}_1, \dots \share{a_{\gamma,j}}_p$ (i.e., one share per DPF key).
The coefficient is defined as follows $a_{\gamma,j} =1$ if $\gamma = \gamma_*$, 0 otherwise.
Finally, it sets a ``correction word'' $W\in(\mathbb{F}_q)^{\nu}$ such that $W + \sum_{i=1}^{q^{p-1}} G(s_{\gamma_*,j}) = e_{\delta_*}\cdot\beta$ (with $e_{\delta}$ a unit vector equal to 1 on index $\delta$, 0 otherwise).
Each key $k_i$ contains the correction word, their share of the coefficients $\share{a_{\gamma,j}}_i$ (for all rows $\gamma$ and all $j\in\{1, \dots, q^{p-1}\}$), and it contains all the seeds $s_{\gamma, j}$ for which $\share{a_{\gamma,j}}_i\neq0$.
This last condition (on the seed inclusion in a key) ensures that there is at least one seed unknown to an adversary composed of $p-1$ out of $p$ parties. 

The $\Eval$ algorithm represents the input $x$ as a tuple $(\gamma, \delta)$, expands the corresponding seeds $s_{\gamma,j}$, multiplies the expanded seeds with the corresponding shared coefficient $\share{a_{\gamma,j}}_i$, sums everything with the correction word $W$, and the share $\share{f(x)}_i$ is on the $\delta$-th index of the sum vector:
$$
\share{f(x)}_i = v[\delta]\text{ with }v = W + \sum_j \share{a_{\gamma,j}}_i \cdot G(s_{\gamma,j})$$

The $\Decode$ algorithm is a basic additive share decoding: $\sum_i \share{f(x)}_i = f(x)$.

As we optimize \cite{boyle_function_2015}, Algorithm \ref{alg:dpf} (describing our scheme) follows roughly the same structure as their scheme.
The \emph{only difference} is the matrix sampling in Lines \ref{lin:array_sample} and \ref{lin:array_sample_bis}.
We then refer to Algorithm \ref{alg:dpf} for more details.

An element could surprise the reader: the number of random seeds $\nu\times q^{p-1}$.
Such a large number is necessary, so an adversary cannot infer information about the secret function based on the distribution of the shares $\share{a_{\gamma,j}}_i$.
These shares can be structured as matrix shares $A_{\gamma}$ with $A_{\gamma}[i,j] = \share{a_{\gamma,j}}_i$.
To preserve function privacy, each matrix must contain all combination of additive shares summing to 1 (if $\gamma=\gamma_*$) or 0 (if $\gamma\neq\gamma_*$).
If the matrices do not contain all possible combinations, an adversary (owning up to $p-1$ keys out of $p$) can recover $\gamma_*$ based on the share distribution \cite{boyle_function_2015}.
Since there exists $q^{p-1}$ combinations of $p$ shares of 0 (resp. of 1) in $\mathbb{F}_q$, the key generator must sample $q^{p-1}$ random seeds.
\section{Our honest-majority scheme}
\label{app:honest-maj-scheme}
The main scalability bottleneck in \cite{boyle_function_2015} lies in the size of the matrices of shares, we improve their scheme by eliminating its dependence on the field size; thanks to the honest-majority assumption ($m<p/2$).
While Boyle et al. \cite{boyle_function_2015} had assumed a dishonest majority ($m<p$), the honest-majority assumption enables us to redesign the matrices, resulting in more compact keys.
Algorithm \ref{alg:dpf} presents our multi-party DPF scheme, and Figure \ref{fig:dpf-struct} an overview of our DPF keys.

Our matrix sampling (described in the function $\text{MatrixOfShares}$ of Algorithm \ref{alg:dpf}) generates a matrix $A$ that distributes shares among all possible combinations of $m+1$ parties out of $p$.
For each combination $S_j$ of $m+1$ parties, the function samples $m+1$ shares of the secret coefficient $a$ (using additive secret sharing) and assigns each share to the cell of $A$; $A[i,j] = \share{a}_i$ if $i\in S_j$, $0$ otherwise.
The secret coefficient $a$ is either set to $0$ or $1$ depending on whether the value $x$ being evaluated matches $\alpha$ (the non-zero point).

\begin{algorithm}[t]
  \caption{Honest-majority DPF scheme adapted from \cite{boyle_function_2015}.}\label{alg:dpf}
  \begin{algorithmic}[1]
    \Function{MatrixOfShares}{$a$}
    \State{Initialize $A$ an $p\times C$ matrix with zeros and the counter $k$ to $1$.}
    \For{each set of parties $S_j$ in the set of all combinations of $m+1$ of $p$}
    \State{Sample $m+1$ shares of the value $a$: $\{\share{a}_{1,i} \dots \share{a}_{m+1,i}\}$.}
    \For{each $i$ in $S_j$}
    {$A[i,j]\leftarrow\share{a}_{k,i}$ and increment $k$.}
    \EndFor
    \EndFor
    \Return A
    \EndFunction

    \Statex
    \Function{$\Gen$}{$\alpha, \beta, p, m, 1^\lambda$}
    \State{Represent $\alpha$ as a pair $\alpha = (\gamma_*, \delta_*)$ with $\gamma_*, \delta_* \in \{0 \dots \nu\}$.}

    \State{Sample $A_1, \dots, A_\nu$ s.t. for all $\gamma \neq \gamma_*$, $A_{\gamma} \leftarrow \text{MatrixOfShares}(0)$.\label{lin:array_sample}}
    \State{Sample $A_{\gamma_*} \leftarrow \text{MatrixOfShares}(1)$.\label{lin:array_sample_bis}}
    \State{Choose randomly and independently $\nu\cdot C$ seeds $s_{1,1}, \dots, s_{\nu,C}\in\{0,1\}^\lambda$.}
    \State{Set the correction words $W \in\mathbb{G}^{\nu}$ s.t. $W + \sum_{j=1}^C G(s_{\gamma_*,j}) = e_{\delta_*}\cdot\beta$.}
    \For{$i \in \{1 \dots p\}, j \in \{1 \dots C\}, \gamma \in \{1 \dots \nu\}$}
    \If{$A_{\gamma}[i,j] \neq 0$}~{$\sigma_{i,\gamma,j}\leftarrow(s_{\gamma,j}, A_{\gamma}[i,j])$.}
    \Else~{$\sigma_{i,\gamma,j}\leftarrow(0,0)$.}
    \Comment{Receives no seed and no coefficient share}
    \EndIf
    \EndFor
    \State{Set $\sigma_{i,\gamma}\leftarrow (\sigma_{i,\gamma,1}|| \dots || \sigma_{i,\gamma, C})$ for all $1\le i \le p$, $1 \le \gamma \le \nu$.}
    \State\Return{$(k_1, \dots, k_p)$ with $k_i = (\sigma_{i,1} || \dots || \sigma_{i,\nu}|| W)$ for all $1\le i \le p$.}
    \EndFunction
    \Statex
    \Function{$\Eval$}{$k_i, x$}
      \State{Represent $x$ as a pair $x = (\gamma, \delta)$ with $\gamma, \delta \in \{0 \dots \nu\}$.}
      \State{Parse $k_i= ((s_{1,1}, A_1[i,1])||\ldots||(s_{1,C}, A_1[i, C])|| \ldots || (s_{\nu, C}, A_\nu[i, C])||W)$.}
      \State\Return{$y_i[\delta]$ with $y_i\leftarrow A_{\gamma}[i,1]\cdot W + \sum_{j=1}^{C} A_{\gamma}[i,j] \cdot G(s_{\gamma,j})$.}
    \EndFunction
    \Statex
    \Function{$\Decode$}{$\share{y}_1, \dots, \share{y}_p$}
    \Return{$\sum_{i=1}^p \share{y}_i$.}
    \EndFunction
  \end{algorithmic}
\end{algorithm}

\begin{figure}[t]
  \centering
  \includegraphics[width=.8\linewidth]{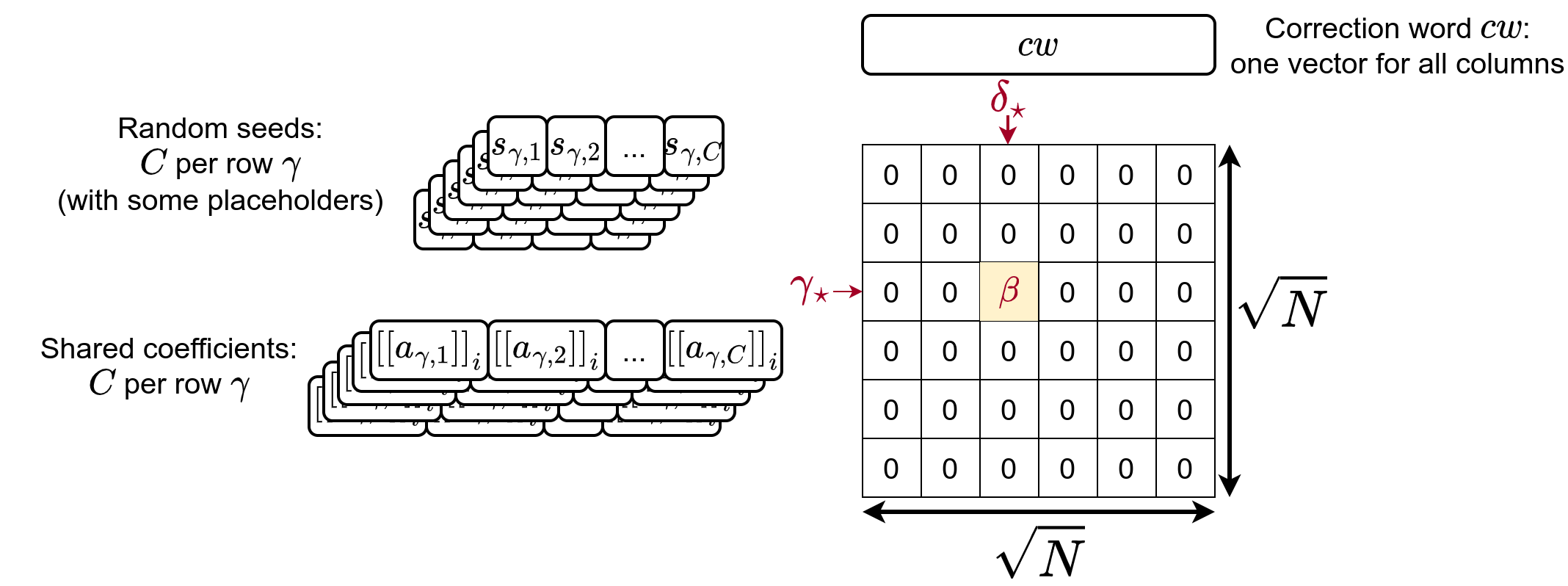}
  \caption{Structure of our DPF keys}
  \label{fig:dpf-struct}
\end{figure}

While in \cite{boyle_function_2015}, each matrix of shares has $p$ rows and $q^{p-1}$ columns, our construction produces matrices with $p$ rows and $\binom{p}{m+1}$ columns.
As in \cite{boyle_function_2015}, we sample one random seed per column.
The $i$-th key contains the $i$-th row of the matrix and includes the $j$-th seed if the corresponding cell $A[i,j]$ is not null.

As in \cite{boyle_function_2015}, it is necessary that (for any given $\gamma$) at least one seed $s_{\gamma,j}$ remains unknown to the adversary.
With all the seeds, they could unmask the correction word $W$ and recover $\beta$.
Our scheme provides each seed to $m+1$ parties, so \emph{under honest majority}, at least one combination of $m+1$ out of $p$ parties contains only honest agents.
Thus, there is at least one seed unknown to an adversary.

Our optimization cannot be extended to dishonest majority.
Indeed, with $m>p/2$, the adversary would know all the seeds because there would be at least one corrupted party in each combination of $m+1$ out of $p$ parties.

Based on these intuitions, we can consider the following security theorem:

\begin{theorem}
    Let $\lambda \in \mathbb{N}$, $N, p \in \mathbb{N}$, then the tuple $(\Gen, \Eval, \Decode)$ as described in Algorithm \ref{alg:dpf} is an FSS scheme for the family of all point functions with $\alpha \in \{1, \dots, N\}$ and any $\beta \in \mathbb{F}_q$.

  Assuming that there exists a secure PRG, then this scheme is correct and private against at most $m$ semi-honest parties with $m < p/2$.
\end{theorem}
\begin{proof}
  As we modified only slightly the scheme of \cite{boyle_function_2015} (i.e., redesigned the matrix of shares based on the honest-majority assumption), our proof follows the same structure as theirs.
  We provide a proof sketch, and refer to \cite{boyle_function_2015} for more details.

  The correctness can be verified by an easy arithmetic exercise considering successively three cases: (1) $\gamma \neq \gamma_*$, (2) $\gamma = \gamma_*$ and $\delta \neq \delta_*$, and (3) $\gamma = \gamma_*$ and $\delta = \delta_*$. Using the $\sqrt{N}\times\sqrt{N}$ grid (illustrated in Figure \ref{fig:dpf-struct}), we represent any input $x$ as $(\gamma, \delta)$ and $\alpha$ as $(\gamma_*, \delta_*)$.

  For privacy, we must show that there exists a simulator that outputs samples from a distribution that is computationally indistinguishable from the distribution of the real DPF keys.
  We propose to study separately: the random seeds $s_{\gamma,j}$, the correction word $W$, and the coefficient shares $A_{\gamma}[i,j]$.

  The simulation is straightforward: for each (simulated) seed, sample a random seed $\widetilde{s_{\gamma,j}}$; for the correction word, sample a random vector $\widetilde{W}$; for the coefficients, for each $S_j$ combination of $m+1$ parties of $p$, for each $i$ in $S_j$, sample a random value and store it in $\widetilde{A_{\gamma}}[i,j]$ (the rest of $\widetilde{A_{\gamma}}$ is null).
  The simulator can return ``simulated'' keys based on these elements.

  Since both the real and simulated seeds are randomly sampled, the simulator’s output distribution ($\widetilde{s_{\gamma,j}}$) is computationally indistinguishable from that induced by the distribution of a single output of $\Gen$.

  The correction word $W$ is a secret vector (i.e., $e_{\delta_*} \cdot \beta$) masked with the output of $\binom{p}{m+1}$ seeded PRGs.
  Remember that a key $k_i$ contains a seed $s_{\gamma,j}$ only if the $i$ is part of the $j$-th combination of $m$ out of $p$ parties.
  So, under honest majority, there is always at least one seed unknown to the adversary controlling $m$ out $p$ parties.
  Hence, the correction word $W$ is computationally indistinguishable from the randomly sampled $\widetilde{W}$, because it is masked with the output of (at least) one PRG seeded with a seed unknown to the adversary \cite{corrigan-gibbs_riposte_2015}.

  Finally, each coefficient is shared between $m+1$ parties, so an adversary controlling $m$ parties cannot distinguish real shares from random values $\widetilde{A_{\gamma}}[i,j]$ provided by the simulator.
\end{proof}

\paragraph{Key size optimization:}
Instead of using a $\sqrt{N} \times \sqrt{N}$ grid, we should use a grid with $\sqrt{N}(C)^{-1}$ rows and $\sqrt{N\cdot C}$ columns, with $C=\binom{p}{m+1}$.
This ``non-square'' grid leverages the fact that each row requires $C$ seeds, and a unique vector $W$ for all columns.
Thanks to this trick, we obtain a better key size: $O(\sqrt{N \cdot C}(\lambda + \log q))$.

\paragraph{Extension to comparison functions}
Boyle et al. \cite{boyle_function_2015} presented a simple adaptation of their DPF to support comparison functions; functions such that $f(x) = \beta$ when $x\le \alpha$, $0$ otherwise.
These schemes have applications notably in machine learning \cite{boyle_secure_2019,kumar_compact_2024}.
We can naively reuse our optimization on this other scheme.
\section{Key size comparison}
\label{sec:comparison}
This section compares our optimized scheme to existing schemes in order to identify asymptotic and practical key size reductions.

We focus our comparison on schemes supporting all possible DPF applications; excluding schemes based on elliptic curves that support a limited number of applications due to their non-linear decoding \cite{kumar_compact_2024} (e.g., do not support PIR).

\paragraph{Asymptotic}
Our scheme provides a key size of $O(\sqrt{N \cdot \binom{p}{m+1}}\log q)$, which is clearly better than the key size of \cite{boyle_function_2015} (i.e., $O(\sqrt{N}q^{\frac{p-1}{2}}\log q)$).

As Bunn et al. \cite{bunn_cnf-fss_2022} built an honest-majority scheme with $O(\sqrt[4]{N})$ key size upon \cite{boyle_function_2015}, we can comment how we distinguish from them.
Their paper does not modify \cite{boyle_function_2015}, but combine it with replicated secret sharing to reduce the dependency on $N$ (i.e., $O(\sqrt[4]{N})$ instead of $O(\sqrt{N})$).
However, their approach worsened the exponential factors already in \cite{boyle_function_2015}; as shown in our benchmark below.
On the contrary, we leverage the honest-majority scheme to avoid exponential factors present in \cite{boyle_function_2015}, but we maintained the same dependence on $N$.

The IT DPF by \cite{bunn_cnf-fss_2022} has a key size comparable to ours: $O(\sqrt{N} \cdot\binom{p}{m+1}\log q)$. However, our approach saves a factor $\sqrt{\binom{p}{m+1}}$ compared to them, yielding significant key size reductions in practice.

While these asymptotic comparisons are informative, they are often insufficient to assess practical performance.
Bunn et al. \cite{bunn_cnf-fss_2022} exemplify this problem: although they substantially reduced the dependence on the domain size $N$, they kept exponential factors without providing any concrete efficiency analysis.
Therefore, we present a detailed comparison based on exact key sizes to offer a more accurate assessment.

\paragraph{Exact}
As we aimed to avoid the exponential factor $q^p$ (for outputs in $\mathbb{F}_q$), we start by studying the dependency on $q$.
Figure \ref{fig:comp-modulus} compares key sizes for varying prime moduli.
Our benchmark includes a curve ``Trivial scheme'' corresponding to the most trivial DPF: sharing the function truth table (i.e., $O(N)$ key size).
This curve serves as baseline to identify impractical solutions.
For any $q>5$, the PRG-based solutions \cite{boyle_function_2015,bunn_cnf-fss_2022} have keys orders of magnitude larger than those of this trivial solution, while the IT DPF of \cite{bunn_cnf-fss_2022} and ours are below.

Recently, Boyle \cite{boyle_function_2022} showed that, for a composite modulus $m=q_1q_2\dots q_l$, we can use the Chinese Remainder Theorem (CRT) to replace $q^p$ with $\sum q_i^p$.  
Figure \ref{fig:comp-modulus-prime} compares key sizes for arbitrary moduli, but the variations make the figure poorly readable.
Instead, Figure \ref{fig:comp-modulus-primorial} compares key sizes for primorial moduli; a primorial is the product of the first $n$ primes.
Primorials are the best case of this CRT trick as they provide composite moduli with the smallest primes possible.

Even with the CRT trick, the existing PRG-based schemes \cite{boyle_function_2015,bunn_cnf-fss_2022} provides key sizes larger than the trivial scheme for any modulus above 210.
As their key sizes are impractical, we exclude these schemes from our other figures to focus on the comparison of our scheme to practical schemes.

Figure \ref{fig:comp-modulus-primorial} also shows that the key size of the IT scheme of \cite{bunn_cnf-fss_2022} grows faster with the modulus than ours.
This phenomenon is explained by the fact that our key size is dominated by components conditioned by a security parameter that is independent of the modulus.

\begin{figure}
    \centering
    \begin{subfigure}{0.32\linewidth}
        \includegraphics[width=\linewidth]{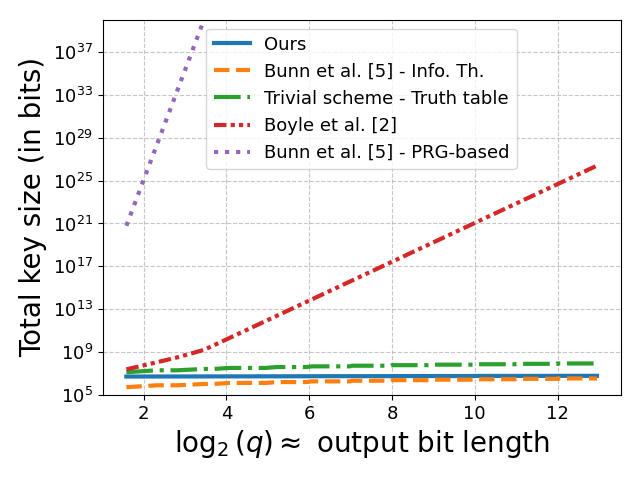}
        \caption{Prime moduli}
        \label{fig:comp-modulus-prime}
    \end{subfigure}
    \hfil
    \begin{subfigure}{0.32\linewidth}
        \includegraphics[width=\linewidth]{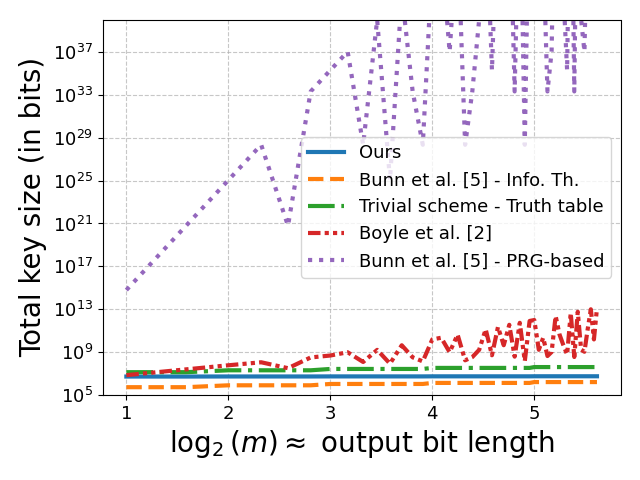}
        \caption{Arbitrary moduli}
        \label{fig:comp-modulus-arbitrary}
    \end{subfigure}
    \hfil
    \begin{subfigure}{0.32\linewidth}
        \includegraphics[width=\linewidth]{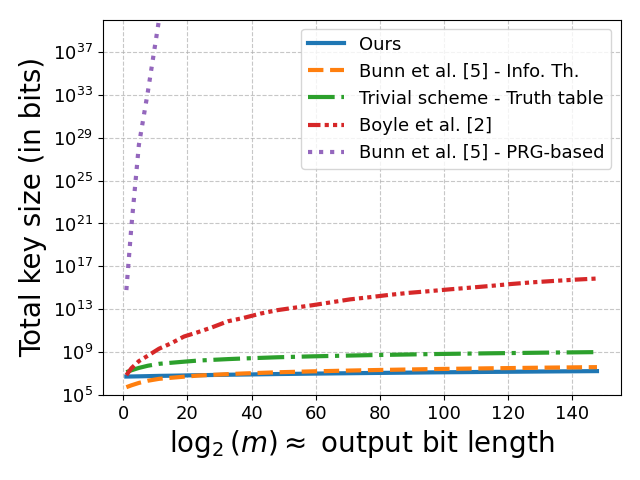}
        \caption{Primorial moduli}
        \label{fig:comp-modulus-primorial}
    \end{subfigure}
    \caption{Key size of various DPF schemes for varying moduli ($p=7$).}
    \label{fig:comp-modulus}
\end{figure}

\begin{figure}
    \centering
    \begin{subfigure}{0.4\linewidth}
        \centering
        \includegraphics[width=.8\linewidth]{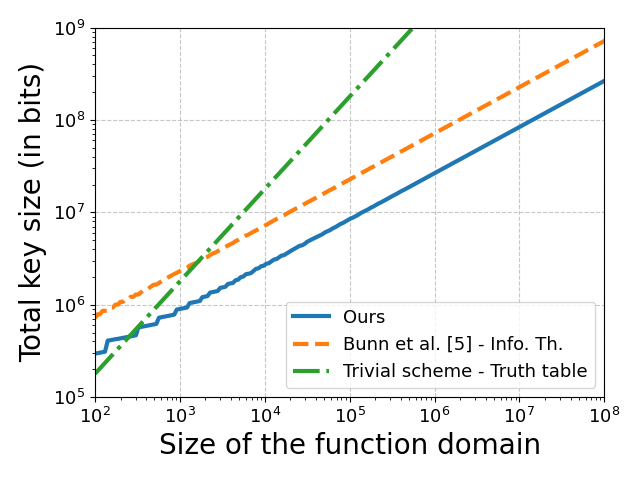}
        \caption{Varying domain size ($p=7$)}
        \label{fig:comp-domain-size}
    \end{subfigure}
    \hfil
    \begin{subfigure}{0.45\linewidth}
        \centering
        \includegraphics[width=.7\linewidth]{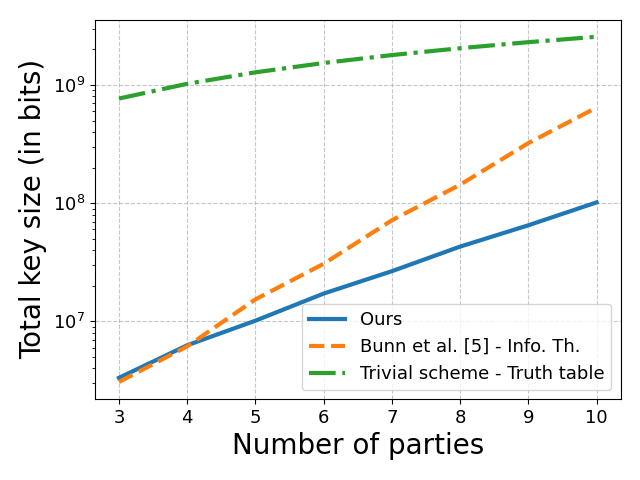}
        \caption{Varying nb. of parties ($N = 10^6$)}
        \label{fig:comp-nb-parties}
    \end{subfigure}
    \caption{Key size of the most efficient DPF schemes.}
\end{figure}

Figure \ref{fig:comp-domain-size} compares the practical schemes for varying  function domain sizes.
\textbf{Our key size is $2.4\times$ smaller than the best existing scheme}.
Moreover, Figure \ref{fig:comp-nb-parties} shows that \textbf{our scheme has a better scaling with the number of parties} thanks to the factor $\sqrt{\binom{p}{m+1}}$ identified in the asymptotic comparison.

\paragraph{Conclusion}
Our optimization based on the honest-majority assumption transformed a PRG-based DPF \cite{boyle_function_2015} with impractical key sizes into the DPF with the smallest key sizes.
Our work proves that, like in two- and three-party schemes, PRG is a promising primitive to build multi-party  DPF with compact keys.


\begin{credits}
    \subsubsection{\ackname}
    This work was supported by the Netherlands Organization for Scientific Research (De Nederlandse Organisatie voor Wetenschappelijk Onderzoek) under NWO:SHARE project [CS.011].
\end{credits}
%
%
%

\bibliographystyle{splncs04}
\bibliography{ref}

\end{document}